\documentclass{article}
\usepackage{fullpage}

\usepackage{algorithmic}
\usepackage{algorithm}
\usepackage{url}
\usepackage{graphicx}
\usepackage{color}
\usepackage{listings}
\usepackage{subfigure}

\usepackage[utf8]{inputenc}
\usepackage{amsmath}
\usepackage{amssymb}

\newtheorem{lemma}{Lemma}

\newenvironment{proof}{\par\noindent{\sc Proof:}\space}{$\Box$\protect\\ \par}

\floatname{algorithm}{Listing}

\usepackage{pgf}
\usepackage{tikz} % tikz graphics
\usetikzlibrary{chains,fit,shapes,arrows,automata,petri,topaths,backgrounds,calc}

\title{The Lock-free $k$-LSM Relaxed Priority Queue\thanks{This paper
    is a full version of a poster with the same title presented at
    PPoPP'15~\cite{Traff15:klsm}.}}

\author{
Martin Wimmer\\
Google\thanks{A large part of this author's work was carried out at and financed by the Vienna University of Technology}\\
\url{wimmerm@google.com}\\
\and
Jakob Gruber\\
Faculty of Informatics\\
Vienna University of Technology\\
\url{gruber@par.tuwien.ac.at}\\
\and
Jesper Larsson Tr\"aff \\
Faculty of Informatics\\
Vienna University of Technology\\
\url{traff@par.tuwien.ac.at}\\
\and
Philippas Tsigas\\
Computer Science and Engineering\\
Chalmers University of Technology\\
\url{tsigas@chalmers.se}\\
}

\begin{document}
\maketitle

\begin{abstract}
Priority queues are data structures which store keys in an ordered
fashion to allow efficient access to the minimal (maximal)
key. Priority queues are essential for many applications, e.g.,
Dijkstra's single-source shortest path algorithm, branch-and-bound
algorithms, and prioritized schedulers.

Efficient multiprocessor computing requires implementations of basic
data structures that can be used concurrently and scale to large
numbers of threads and cores. Lock-free data structures promise
superior scalability by avoiding blocking synchronization primitives,
but the \emph{delete-min} operation is an inherent scalability
bottleneck in concurrent priority queues. Recent work has focused on
alleviating this obstacle either by batching operations, or by
relaxing the requirements to the \emph{delete-min} operation.

We present a new, lock-free priority queue that relaxes the
\emph{delete-min} operation so that it is allowed to delete \emph{any}
of the $\rho+1$ smallest keys, where $\rho$ is a runtime configurable
parameter.  Additionally, the behavior is identical to a non-relaxed
priority queue for items added and removed by the same thread. The
priority queue is built from a logarithmic number of sorted arrays in
a way similar to log-structured merge-trees.  We experimentally
compare our priority queue to recent state-of-the-art lock-free
priority queues, both with relaxed and non-relaxed semantics, showing
high performance and good scalability of our approach.
\end{abstract}

\noindent
\textbf{Keywords} Task-parallel programming, priority-queue, concurrent data
structure relaxation, parallel single-source shortest path

\section{Introduction}

A priority queue is a data structure for maintaining a set of
\emph{keys} (potentially stored along with some data) such that the
minimum (or maximum) key can be efficiently accessed and removed. In
its simplest form, a priority queue supports insertion (\emph{insert})
of new keys, and finding and deleting a minimal key
(\emph{delete-min}). Priority queues may additionally support deletion
of arbitrary keys and decreasing (or increasing) the value of a key,
as well as operations to meld and split queues. Priority queue
operations can generally be performed in $O(\log n)$ time, some in
$O(1)$ time, but either insert or delete-min must take $\Omega(\log
n)$ time~\cite{Thorup2007-sn}.

Parallel and concurrent priority queues have been the subject of
research since the
1980s~\cite{ayani1990lr,BiswasBrowne87,Traff98:ppq,deo1992parallel,karp1993randomized,luchetti1993some,olariu1991optimal,prasad1995parallel,sanders1998randomized}.
While early efforts have focused mostly on parallelizing heap
structures~\cite{hunt1996efficient}, recent implementations of
priority
queues~\cite{alistarhspraylist,herlihy2012art,linden2013skiplist,shavit2000skiplist,SundellTsigas05}
were often based SkipLists~\cite{pugh1990skip}.  The randomized
priority queue~\cite{LiuSpear12} is based on a tree of sorted arrays.

It is commonly believed that \emph{lock-} and \emph{wait-free} data
structures provide the best scalability for multiprogrammed
environments, because stalling threads cannot block progress of (all)
other threads.  However, for priority queues, the delete-min operation
remains an inherently sequential scalability bottleneck.  Recent
approaches attempt to alleviate this obstacle by batching and
elimination~\cite{Calciu14}, or by relaxing the linearization
requirements for insertions~\cite{Wimmer2014-ct,Wimmer2014-ym} and
deletions~\cite{alistarhspraylist}. A very appealing, randomized
lock-based approach was recently described in~\cite{Sanders14}.

In this paper we present a new, lock-free concurrent priority queue
built from a logarithmic number of sorted arrays of keys, similar
to the \emph{log-structured merge-trees} used in databases~\cite{ONeil1996-kz}.
It relaxes linearization requirements on \emph{insertions}, thus
allowing each thread to batch together up to $k$ insert operations
before it is required to linearize them wrt.\ insertions by
other threads. The \emph{delete-min} operation is also relaxed so that
it may delete and return any of the $k+1$ smallest keys visible to all
threads. The parameter $k$ can be configured at run-time,
and can even be changed on a per-key basis.

Combining both relaxations, our $k$-LSM priority queue is allowed to
ignore up to $\rho = Tk$ keys inserted into the priority queue at any
time ($T$ being the number of threads), but never more. Despite the
relaxations, the priority queue preserves local semantics per thread,
so that keys inserted and deleted by the same thread will always be
deleted in the correct order. We provide a \texttt{C++} implementation
of this data structure. Experiments show high single thread
performance, and very good scalability when choosing a reasonably
large value for $k$.

The paper is structured as follows: we discuss general ordering
semantics for concurrent data structures in
Section~\ref{sec:ordering_semantics}. In Section~\ref{sec:lsm} we
explain the sequential LSM algorithm used as basis for our concurrent
priority queue. Our concurrent algorithm and its implementation is
described in Section~\ref{sec:klsm}. Section~\ref{sec:correctness}
establishes correctness and progress guarantees. Finally, we
experimentally compare the $k$-LSM priority queue to other
state-of-the-art priority queue implementations using a synthetic
benchmark, and a concurrent variant of Dijkstra's algorithm for
single-source shortest paths.

\section{Ordering semantics}
\label{sec:ordering_semantics}

Ordering semantics provide trade-offs between scalability and
linearizability guarantees on \emph{update operations} (\emph{insert}
and \emph{delete-min}) as well as \emph{read operations}
(\emph{find-min}).

\paragraph{Global ordering semantics}
\label{sec:global_ordering_semantics}
The strictest possible semantics for concurrent ordered containers is
to linearize all update operations by all threads with regard to each
other. For priority queues this will typically lead to high contention
on concurrent \emph{delete-min} operations, since all threads will
attempt to remove the same item, and only one can succeed.

\paragraph{Local ordering semantics}
\label{sec:purely_local}
The other end of the spectrum occurs when threads maintain their own,
local copy of the data structure. When a thread accesses the
thread-local copy of another thread, its operations are linearized
with regard to operations of other threads accessing the same
copy. Operations on distinct thread-local data structures are not
linearized with regard to each other, and no global guarantees can be
given. 
Purely local semantics can be found, e.g., in
work-stealing deques~\cite{Arora2001-fr}.

\paragraph{Quantitative or $\rho$-relaxation}
\label{sec:k-relaxation}
Afek et al.~\cite{Afek2010-lh} introduced an alternative consistency
model called \emph{quasi linearizability}, which extends
linearizability by allowing operations to occur out of a correct
linearizable order.  Quasi linearizability imposes an upper bound on
the \emph{distance} each operation is allowed to have from a correct
linearizable operation order.  Quasi linearizability has been used as
a consistency condition for FIFO
queues~\cite{Afek2010-lh,Basin2011-xs}, but is not restricted to
these. Later work by Henzinger et al.~\cite{Henzinger2013-vr} provided
a closely related model called \emph{quantitative relaxation}.

In previous work~\cite{Wimmer2014-ct,Wimmer2014-ym}, we introduced the
term \emph{$\rho$-relax\-ation} to describe quantitative relaxation on
data structures, where an upper bound, $\rho$, can be given on the
number of items that can be \emph{skipped} on data structure
accesses. We say an item is \emph{skipped} whenever an operation on
the data structure is supposed to return an item according to global
ordering semantics, but instead returns the item it would return if
the other item did not exist. This includes the case where a
null-value is returned by a data structure, making it look empty,
since all stored items were skipped.

We distinguish two types of $\rho$-relaxation: \emph{temporal} and
\emph{structural} as shown in
Figure~\ref{fig:temp_vs_struct_rho}. Temporal $\rho$-relaxation is
based on the recency of items, and closely related to quasi
linearizability. A temporally $\rho$-relaxed data structure is only
allowed to skip the $\rho$ most recently added items. Temporal
$\rho$-relaxation can be applied \emph{globally}, so that only the
$\rho$ most recent items added by \emph{any} thread can be skipped, or
\emph{locally}, where the $\rho$ most recent items by \emph{each}
thread can be skipped. In contrast, structural $\rho$-relaxation is
only concerned with the number of items that are allowed to be skipped
at any point in time regardless of their recency.

Note that in a $\rho$-relaxed data structure all items still need to
be accessible by all threads at any point in time after their
insertion to keep the $\rho$-relaxation guarantees in case of stalling
threads. The relaxation allows to deliberately ignore a certain number
of items to improve scalability.

\begin{figure}
\centering
\subfigure[\texttt{\{A,B,C\}}]{
\begin{tikzpicture}[start chain=going right,
    box/.style={
		draw,rectangle,
		text centered,
		minimum width=5mm,
		minimum height=5mm,
		font=\footnotesize,
		fill=black!10
	},
	node distance=0pt]

\node[box] at (0, 0) {$A$};
\node[box] at (0, 5mm) {$B$};
\node[box] at (0, 10mm) {$C$};

\draw (5mm,12.5mm) -- node[sloped, above] {$\rho_t$} (5mm, 2.5mm);
\draw (4mm,12.5mm) -- (6mm,12.5mm);
\draw (4mm,2.5mm) -- (6mm,2.5mm);

\draw (-5mm,2.5mm) -- node[sloped, above] {$\rho_s$} (-5mm, 12.5mm);
\draw (-4mm,12.5mm) -- (-6mm,12.5mm);
\draw (-4mm,2.5mm) -- (-6mm,2.5mm);
\end{tikzpicture}
}
\subfigure[\texttt{push($D$)}]{
\begin{tikzpicture}[start chain=going right,
    box/.style={
		draw,rectangle,
		text centered,
		minimum width=5mm,
		minimum height=5mm,
		font=\footnotesize,
		fill=black!10
	},
	node distance=0pt]

\node[box] at (0, 0) {$A$};
\node[box] at (0, 5mm) {$B$};
\node[box] at (0, 10mm) {$C$};
\node[box] at (0, 15mm) {$D$};

\draw (5mm,17.5mm) -- node[sloped, above] {$\rho_t$} (5mm, 7.5mm);
\draw (4mm,17.5mm) -- (6mm,17.5mm);
\draw (4mm,7.5mm) -- (6mm,7.5mm);

\draw (-5mm,7.5mm) -- node[sloped, above] {$\rho_s$} (-5mm, 17.5mm);
\draw (-4mm,17.5mm) -- (-6mm,17.5mm);
\draw (-4mm,7.5mm) -- (-6mm,7.5mm);
\end{tikzpicture}
}
\subfigure[\texttt{pop($C$)}]{
\begin{tikzpicture}[start chain=going right,
    box/.style={
		draw,rectangle,
		text centered,
		minimum width=5mm,
		minimum height=5mm,
		font=\footnotesize,
		fill=black!10
	},
	node distance=0pt]

\node[box] at (0, 0) {$A$};
\node[box] at (0, 5mm) {$B$};
\node[box, dotted, fill=none] at (0, 10mm) {$C$};
\node[box] at (0, 15mm) {$D$};

\draw (5mm,17.5mm) -- node[sloped, above] {$\rho_t$} (5mm, 7.5mm);
\draw (4mm,17.5mm) -- (6mm,17.5mm);
\draw (4mm,7.5mm) -- (6mm,7.5mm);

\draw (-5mm,2.5mm) -- node[sloped, above] {$\rho_s$} (-5mm, 17.5mm);
\draw (-4mm,17.5mm) -- (-6mm,17.5mm);
\draw (-4mm,2.5mm) -- (-6mm,2.5mm);
\end{tikzpicture}
}
\caption{Temporal vs. structural $\rho$ relaxation on a stack for
  $\rho=2$. The items that can be relaxed by structural
  $\rho$-relaxation are marked with $\rho_s$, by temporal with
  $\rho_t$. After the pop, temporal $\rho$-relaxation is not allowed
  to skip $B$, since two items were added to the stack after $B$, even
  though $C$ was deleted in the meantime.}
\label{fig:temp_vs_struct_rho}
\end{figure}
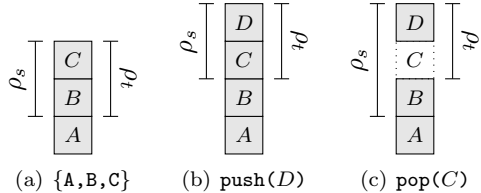

\section{Log-structured Merge-tree}
\label{sec:lsm}
We now present our new priority queue based on \emph{log-structured
  merge-trees} (LSM). Log-structured merge-trees~\cite{ONeil1996-kz}
were introduced in the database community, where they are used as a
disk-based index structure. They consist of a logarithmic number of
sorted arrays and provide high efficiency for applications with large
amounts of item retrievals and removals and rare lookups. The data
structure is appealing for priority queue implementations because that
finding the minimum (maximum) item can be done much faster than other
lookups. Furthermore, because the data structure is designed to reduce
the number of disk accesses and has a low number of non-contiguous
disk accesses it can likewise be implemented in a cache-efficient
manner\footnote{The LSM presented here was invented independently,
  based on requirements for concurrent relaxed priority queues; only
  later we discovered the relation to the data structure used in the
  database community.}.

A log-structured merge-tree (LSM) priority queue operates on a
logarithmic number of sorted arrays, which we call \emph{blocks}, as
depicted in Figure~\ref{fig:lsm}. Each block is assigned a
\emph{level}, and a block with level $l$ can store a sorted array of
$n$ keys, where $2^{l-1} < n \leq 2^l$. To guarantee logarithmic
lookup time for retrieving a minimal key, the LSM allows at most one
block of each level at any point in time. If this property is violated
by insertion or deletion of a key, the two blocks with the same level
will be merged into one block with the next-higher level repeatedly
until the property is fulfilled again.

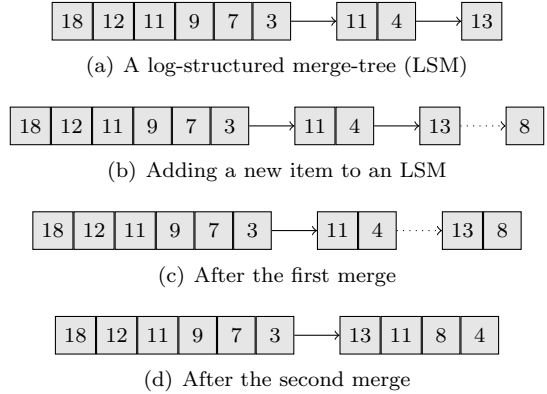
\begin{figure}
\centering
\subfigure[A log-structured merge-tree (LSM)]{
\begin{tikzpicture}[start chain=going right,
    basebox/.style={
		draw,rectangle,
		text centered,
		minimum width=5mm,
		minimum height=5mm,
		font=\footnotesize,
		fill=black!10
	},
	box/.style={
		on chain,
		basebox
	},
	every join/.style={->},
	node distance=0pt]

\foreach \id/\value in {1/$18$, 2/$12$, 3/$11$, 4/$9$, 5/$7$, 6/$3$}
{
  \node[box] (node d\id) {\value};
}

\node[node distance=6mm, box, join] (node 9) {$11$};
\node[box] (node 10) {$4$};

\node[node distance=6mm, box, join] (node 11) {$13$};

\end{tikzpicture}
\label{fig:lsm}
}

\subfigure[Adding a new item to an LSM]{
\begin{tikzpicture}[start chain=going right,
    basebox/.style={
		draw,rectangle,
		text centered,
		minimum width=5mm,
		minimum height=5mm,
		font=\footnotesize,
		fill=black!10
	},
	box/.style={
		on chain,
		basebox
	},
	every join/.style={->},
	node distance=0pt]

\foreach \id/\value in {1/$18$, 2/$12$, 3/$11$, 4/$9$, 5/$7$, 6/$3$}
{
  \node[box] (node d\id) {\value};
}

\node[node distance=6mm, box, join] (node 9) {$11$};
\node[box] (node 10) {$4$};

\node[node distance=6mm, box, join] (node 11) {$13$};

%\tikzstyle[every join/.style={->, dotted}];
\node[node distance=6mm, box] (node 12) {$8$};
\draw[->, dotted] (node 11) -- (node 12);
\end{tikzpicture}
\label{fig:lsm_push_merge_1}
}

\subfigure[After the first merge]{
\begin{tikzpicture}[start chain=going right,
    basebox/.style={
		draw,rectangle,
		text centered,
		minimum width=5mm,
		minimum height=5mm,
		font=\footnotesize,
		fill=black!10
	},
	box/.style={
		on chain,
		basebox
	},
	every join/.style={->},
	node distance=0pt]

\foreach \id/\value in {1/$18$, 2/$12$, 3/$11$, 4/$9$, 5/$7$, 6/$3$}
{
  \node[box] (node d\id) {\value};
}

\node[node distance=6mm, box, join] (node 9) {$11$};
\node[box] (node 10) {$4$};

\node[node distance=6mm, box] (node 11) {$13$};
\node[box] (node 12) {$8$};
\draw[->, dotted] (node 10) -- (node 11);
\end{tikzpicture}
\label{fig:lsm_push_merge_2}
}

\subfigure[After the second merge]{
\begin{tikzpicture}[start chain=going right,
    basebox/.style={
		draw,rectangle,
		text centered,
		minimum width=5mm,
		minimum height=5mm,
		font=\footnotesize,
		fill=black!10
	},
	box/.style={
		on chain,
		basebox
	},
	every join/.style={->},
	node distance=0pt]

\foreach \id/\value in {1/$18$, 2/$12$, 3/$11$, 4/$9$, 5/$7$, 6/$3$}
{
  \node[box] (node d\id) {\value};
}

\node[node distance=6mm, box, join] (node 9) {$13$};
\node[box] (node 10) {$11$};
\node[box] (node 11) {$8$};
\node[box] (node 12) {$4$};
\end{tikzpicture}
\label{fig:lsm_push_merge_4}
}

\caption{A log-structured merge-tree (LSM)}
\label{fig:lsm_merge}
\end{figure}

An \emph{insertion} is shown in Figure~\ref{fig:lsm_merge}. A single
key block at level $0$ is first created and appended to the
LSM. Merges are then performed from the end of the list until no two
blocks with the same level exist.  The \emph{find-min} operation
returns the minimal key in the LSM. Since blocks are sorted, the
minimal key of each block can be accessed in constant time, and the
minimum over all minimal block keys must be returned. The
\emph{delete-min} operation works similarly, but also removes the
minimal key from the LSM. If a block contains too few elements for its
level after removal, it is shrunk to the next-smaller level and merged
with another block of same level if necessary. Thus, the amortized
complexity of all operations is $O(\log n)$.

\section{The {\large \textit{$k$}}-LSM priority queue}
\label{sec:klsm}
The $k$-LSM priority queue consists of two separate components, the
\emph{shared $k$-LSM} and the \emph{distributed LSM priority
  queue}. Both can also be used as a standalone priority queue, but
each has drawbacks that make it beneficial to combine it with the
other.

The shared $k$-LSM priority queue, described in
Section~\ref{sec:shared_klsm}, can provide good guarantees on the
ordering semantics, but suffers from a sequential bottleneck on
insertions and various maintenance operations. The frequency of both
types of operations can be reduced by roughly a factor of $k$ by
inserting sorted blocks of $k$ keys instead of single keys. We use the
distributed LSM priority queue for batching keys together.

The distributed LSM priority queue, described in
Section~\ref{sec:dist_lsm}, does not suffer from these sequential
bottlenecks, because synchronization is kept low. On the other hand,
it has only local ordering semantics and is also wasteful in memory
usage. To achieve stronger ordering guarantees and good space bounds
we combine it with the shared $k$-LSM priority queue, and bound the
size of the distributed LSM by the parameter $k$. Keys from the
distributed LSM are transferred to the shared $k$-LSM priority queue
whenever the bounds are reached, as described in
Section~\ref{sec:combine}.  We first assume that a garbage collector
is available, but describe in Section~\ref{sec:mem_man} how memory
management can be done manually on systems without garbage
collection. Extensions to the data structure are discussed in
Section~\ref{sec:extensions}.

\paragraph{External interface}
\label{sec:inter}
The external interface of our priority queue consists of two
functions: \texttt{insert} and \texttt{try\_delete\_min}. The
\texttt{insert} function inserts a single key into the priority queue
and always succeeds; \texttt{try\_delete\_min} attempts to find the
minimal key (according to the ordering semantics) and delete it. This
function returns a flag indicating whether the operation succeeded (a
key was found and deleted) and the minimal key (on success). The
operation will fail if the priority queue is empty, and may spuriously
fail as long as it is guaranteed that a key will eventually be
returned given enough attempts. Implementations of both functions are
shown in Section~\ref{sec:combine}.

The interface can be extended by an analogous
\texttt{try\_\-find\_\-min} function, and a \texttt{size} function.
The latter returns the number of keys in the priority queue and is
allowed to be off by up to $\rho$ from the actual value, where $\rho$
is the relaxation parameter explained in
Section~\ref{sec:correctness}.

\paragraph{Shared components}
\label{sec:shared_comp}
Each key in the priority queue is wrapped in a small \texttt{Item}
structure. The blocks of an LSM consist of pointers to \texttt{Item}
instances, and more than one pointer to an \texttt{Item} is allowed to
exist.

An \texttt{Item} stores the \texttt{key} alongside a boolean
\texttt{flag}. Items are logically deleted by performing an atomic
\texttt{test\_and\_set} on \texttt{flag}, wrapped in the member
function \texttt{Item::take}. Pointers to logically deleted
\texttt{Item} instances are lazily cleaned out of blocks whenever
blocks are resized or merged.

\begin{algorithm}
\begin{lstlisting}[mathescape=true,columns=flexible,escapechar=---]
struct Block {
	atomic<int> filled;
	int level;
	atomic<Item*>* items;

	Block(int $l$):filled(0), level($l$) {
		items = new atomic<Item*>[$2^l$];
	}

	void append(Item* item) {
		// Only copy items that are not logically deleted
		if (!item->flag) items[filled++] = item;
	}

	Block* copy(int $l$) {
		Block* nb = new Block($l$);
		for (int i = 0; i<filled; ++i) nb->append(items[i]);
		return nb;
	}

	void merge_in(Block* b1, Block* b2); // not shown

	Block* shrink() {
		int $f$ = filled;
		// Check whether items have been logically deleted
		while ($f$ > 0 && !items[$f-1$]->flag) --$f$;

		// Check whether we still use the correct level
		int $l$ = level;
		while ($l>0$ && $f \leq 2^{l-1}$) --$l$;

		// It is necessary to shrink the block
		if ($l$ < level) {
			// Recurse, since copy may clean out items
			return copy($l$)->shrink();
		}
		filled = $f$;
		return this;
}};
\end{lstlisting}
\caption{Pseudocode for the \texttt{Block} data structure, which
  stores a sorted array of \texttt{Item} pointers.}
\label{lst:block}
\end{algorithm}

The \texttt{Block} data structure stores items in decreasing key order
and is shown in Listing~\ref{lst:block}. The level of a block is the
base~$2$ logarithm of the size of its \texttt{Item} array. The member
variable \texttt{filled} stores the actual number of items in the
block. The \texttt{append} method adds a single item to the end of the
block, but only if it has not been logically deleted.  The
\texttt{copy} method copies the block into a new block of given level,
and uses the \texttt{append} method to filter out all logically
deleted items.  The \texttt{merge\_in} method (code not shown)
performs a standard two-way merge of the given blocks and stores the
result using \texttt{append}.  Finally, the \texttt{shrink} method
scans the end of the block for logically deleted items, decrements
\texttt{filled} and if necessary copies the items into a smaller
block.

\subsection{Shared {\large $k$}-LSM priority queue}
\label{sec:shared_klsm}

The shared $k$-LSM priority queue uses a fairly straightforward
parallelization idea: Blocks in the LSM are stored as an array of
pointers to blocks, sorted by level. This array is accessible to all
threads by a single pointer. Each update to the priority queue will
atomically replace the pointer to the array by a pointer to a new
array containing all modifications using a compare-and-swap.

This copy-on-write mechanism is safe since blocks are not modified
once they have been added to the array with the exception of the
member variable \texttt{filled}, which is modified in the
\texttt{shrink} method. While there can be data races on accesses to
\texttt{filled}, all possible outcomes will lead to the block being in
a valid state, since the values being written to \texttt{filled} may
only be larger than the correct value, but never smaller. Merging and
shrinking will always modify new blocks and not modify an existing
block shared with other threads.

\paragraph{Bottlenecks}
There are two main sequential bottlenecks in this implementation:
updating the shared block array, and deleting the minimal
item. Updates to the block array can be reduced by doing bulk
insertions. As a side-effect this leads to the average block being
larger, thus also reducing the frequency of shrinking and merging
operations. We bulk insertions by batching together inserted items
using the distributed LSM priority queue presented in
Section~\ref{sec:dist_lsm}.

The other sequential bottleneck of the shared $k$-LSM priority queue,
deleting a minimal item, is inherent to priority queues. We reduce its
effects by relaxing a \texttt{delete\_min} operation to take any of
the $k+1$ smallest keys in the priority queue, instead of only the
minimal one. The selection is performed uniformly at random, and falls
back to taking the minimal item in the same block if the selected item
was already deleted by another thread. In addition, a thread will
always first search for its locally inserted minimal key and per
default delete this key it if is one of the $k+1$ smallest keys. This
ensures that the shared $k$-LSM priority queue fulfills local ordering
semantics, as described in Section~\ref{sec:purely_local}.

\paragraph{Implementation}
The array, which is atomically replaced whenever a structural update
to the LSM is made, is stored in a data structure called
\texttt{BlockArray} as shown in Listing~\ref{lst:block_array}.  This
also contains an array of pivot indices for each block, which separate
keys guaranteed to be among the smallest $k+1$ keys from other keys.

The methods \texttt{insert}, \texttt{consolidate} and
\texttt{calculate\_pivots} modify a \texttt{BlockArray} and may thus
only be called if it is not visible to other threads. The
\texttt{insert} method adds a block to the \texttt{BlockArray} at its
correct level position, and calls \texttt{consolidate} to ensure that
the levels of blocks in the array are strictly decreasing.

\begin{algorithm}
\begin{lstlisting}[mathescape=true,columns=flexible,escapechar=---]
struct BlockArray {
	Block* blocks[MAX_LEVELS];
	int k_pivots[MAX_LEVELS];
	int size;

	void insert(Block* block);
	bool consolidate();
	void calculate_pivots();
	BlockArray* copy();

	Item* find_min() {
		int total = 0;
		// Find out how many elements we can choose from
		for (int i = 0; i < size; ++i) {
			total += blocks[i]->filled - k_pivots[i];
		} // We can assume that total > 0

		// Select one item uniformly at random
		int r = rand_int(0, total - 1);

		// Find block for r
		for (int i = 0; i < size; ++i) {
			int range = blocks[i]->filled - k_pivots[i];
			if (range <= r) r -= range;
			else {
				// Selected element is in this block
				if (r != range - 1) {
					Item* item = blocks[i]->items[k_pivots[i]+r];
					// Check whether not deleted
					if(!item->flag) return item;
				}
				// Fall back to minimal element in block
				return blocks[i]->items[blocks[i]->filled - 1];
}}}};
\end{lstlisting}
\caption{The struct for storing an array of blocks.}
\label{lst:block_array}
\end{algorithm}

The \texttt{consolidate} method is responsible for shrinking blocks
and performing merges if necessary. It performs two passes on the
\texttt{BlockArray}. In the first pass, blocks are scanned beginning
with the smallest block and shrunk if necessary. If a block after
shrinking has a smaller level than its successor in the
\texttt{BlockArray} it is merged with its successor. In the second
pass, the array is compacted by getting rid of null pointers and empty
blocks. The return value signifies whether a merge occurred during
consolidation.

The \texttt{calculate\_pivots} method selects a pivot value, which is
one of the $k+1$ smallest keys in the LSM, and writes the offset of
the first key less or equal to the pivot key in each block into the
array \texttt{k\_pivots}. The \texttt{copy} method creates an
identical copy of a \texttt{BlockArray}, by simply copying all data
into a new instance. The \texttt{find\_min} method randomly selects
one out of the up to $k$ items with smallest key. If that item is not
marked as deleted it is returned, otherwise the algorithm falls back
to the minimal item in the block.

\begin{algorithm}
\begin{lstlisting}[mathescape=true,columns=flexible,escapechar=---]
struct SharedKLSM {
	atomic<BlockArray*> shared;
	thread_local BlockArray* observed;
	thread_local BlockArray* snapshot;

	void refresh_snapshot(); {
		observed = shared;
		snapshot = observed->copy(observed->level);
	}

	bool push_snapshot() {
		if (shared.cas(observed, snapshot)) return true;
		return false;
	}

	void insert(Block* block); // not shown

	Item* find_min() {
		while (true) {
			if (shared != observed) refresh_snapshot();
			if (snapshot == nullptr) return nullptr;

			Item* item = snapshot->find_min();
			if (item->flag) {
				// Consolidate and push if merges occurred
				bool push = snapshot->consolidate();
				if (snapshot->size == 0) {
					snapshot = nullptr;
					push = true;
				}
				if (push) push_snapshot();
			} else return item;
}}};
\end{lstlisting}
\caption{The shared $k$-LSM priority queue.}
\label{lst:shared_klsm}
\end{algorithm}

The implementation of the shared $k$-LSM priority queue itself is
found in Listing~\ref{lst:shared_klsm}. It mainly consists of a
pointer to an instance of \texttt{BlockArray} called \texttt{shared},
which is shared between all threads, and two thread-local pointers to
\texttt{BlockArray}. Each thread uses the thread-local pointer
\texttt{snapshot} to store a consistent snapshot of
\texttt{shared}. Since this snapshot is private to each thread, it is
also used as a staging area to prepare an update to
\texttt{shared}. The thread-local pointer \texttt{observed} stores a
pointer to the instance of \texttt{shared} that was used to create the
snapshot. All three pointers are initialized to
null-pointers.

The \texttt{refresh\_snapshot} method updates the snapshot by first
copying the \texttt{shared} pointer to \texttt{observed} and then
storing a copy of the referenced \texttt{BlockArray} in
\texttt{snapshot}. The \texttt{push\_snapshot} method is responsible
for updating the sha\-red array after the private snapshot was
modified. This operation may fail if the shared array was modified
after the snapshot was taken. The \texttt{insert} method inserts a
block into \texttt{snapshot}, consolidates, and then attempts to push
the snapshot. Since this may fail due to the shared array being
modified, it will retry on a new snapshot until it succeeds.

The \texttt{find\_min} method finds and returns the minimal key from
the snapshot. If this key was already marked as deleted by another
thread, the array is consolidated. This may trigger a merge operation
or result in an empty snapshot. In both cases, an attempt is made to
replace \texttt{shared} with the modified snapshot. Failure to update
\texttt{shared} means that another thread already succeeded in
consolidating it and thus no retry is necessary.

\paragraph{Local ordering semantics}
To simplify the presentation, the implementation shown here does not
support local ordering semantics, which require that a thread will
never skip items that it inserted itself. To support this we use a
Bloom filter with each \texttt{Block} that identifies all threads that
contributed items to the block. The \texttt{find\_min} operation in
\texttt{BlockArray} then checks the Bloom filter of each block for the
current thread ID. The minimal key out of the blocks with the current
thread ID in the Bloom filter is then compared with the key of the
randomly selected item, and the smaller item returned. We use 64-bit
Bloom filters with two hash-values obtained by tabular hashing. Since
the Bloom filters are only updated when two blocks are merged, no
synchronization mechanism is necessary (the block being merged into is
not shared with other threads until the merge is completed and the
Bloom filter updated).  Our implementation used with the benchmarks
supports local ordering semantics.

\subsection{Distributed LSM priority queue}
\label{sec:dist_lsm}

The distributed LSM priority queue uses a parallelization idea
inspired by work-stealing: each thread maintains its own priority
queue, and inserts and deletes keys from it. Only when a thread-local
priority queue is empty, that thread will try to access keys stored in
the priority queue of another thread. For this we rely on a scheme
called \emph{spying}~\cite{Wimmer2013-kr}: a thread scans the priority
queue of a randomly selected thread and copies the pointers to items
found in that priority queue into its own priority queue. Spying is
similar to work-stealing, the difference being that it is
\emph{non-destructive}: it does not remove items from the victim's
priority queue. This is necessary to preserve local ordering
semantics.

The main advantage of the distributed LSM priority queue is
scalability: it is essentially embarassingly parallel, and the
potential parallelism increases with the number of threads inserting
items.  The main disadvantage is its lack of global \emph{delete-min}
ordering guarantees.  To provide ordering guarantees (when required),
it is necessary to combine the distributed LSM data structure with
some other data structure.

\paragraph{Implementation}
In the distributed LSM priority queue, each thread has its own
instance of the priority queue. A shared array storing pointers to all
available instances of the priority queue is used for \emph{victim}
selection in the spying operation.

\begin{algorithm}
\begin{lstlisting}[mathescape=true,columns=flexible,escapechar=---]
struct DistLSM {
	atomic<Block*> blocks[MAX_LEVELS];
	int size;

	void insert(Item* item) {
		Block* b = new Block(0);
		b->append(item);

		int i = size;
		// Note: Merge is non-destructive. Old blocks stay
		// available throughout the loop
		while (i > 0 && blocks[i-1]->level <= b->level) {
			Block* b2 = new Block(b->level + 1);
			b2->merge_in(blocks[i-1], b);
			b = b2->shrink();
			--i;
		}
		// Only now replace largest of old blocks with block
		// containing all merged items
		blocks[i] = b;
		// Now redundant blocks become inaccessible
		size = i+1;
	}

	void consolidate(); // not shown
	Item* find_min();

	bool spy(DistLSM* victim) {
		for (int i = 0; i < victim->size; ++i) {
			Block* b = victim->blocks[i];
			int $l$ = b->level;
			if (b != nullptr && (size == 0 || 
                        $l$ < blocks[size-1]->level)) {
				blocks[size++] = b->copy($l$);
			}
		}
		return size != 0;
}};
\end{lstlisting}
\caption{The distributed LSM priority queue.}
\label{lst:dist_lsm}
\end{algorithm}

The implementation of such a thread-local LSM priority queue is shown
in Listing~\ref{lst:dist_lsm}. It consists of an array with pointers
to blocks, and a \texttt{size} integer giving the number of blocks in
the array. Blocks are stored in strictly decreasing block level
order. The \texttt{insert} method creates a block of level $0$
containing the new item. It then performs merges with other blocks in
the LSM starting from the back until no block of same level is
found. The old (non-merged) blocks stay accessible throughout the
merge. The block containing all merged blocks is made visible
replacing the old blocks once all merges have been performed. All
items stay accessible throughout the merge, but may be encountered
twice by spying threads.

The \texttt{consolidate} method scans the block array for blocks
needing to be shrunk, and performs block merges if necessary. It will
only remove references to blocks being consolidated, after the
consolidated blocks are made available. This ensures that a
\texttt{spy} operation has the chance to find all items stored in the
LSM of its victim at all times. The \texttt{find\_min} method behaves
similarly to its counterpart in the sequential LSM data structure.

The \texttt{spy} method scans through the blocks of the victim and
copies the blocks it finds. Since the underlying array at the victim
may be modified during the scan, care is taken only to copy blocks
that will not violate the invariant that blocks are stored in strictly
decreasing order. It is not necessary for the blocks obtained by
\texttt{spy} to be a consistent snapshot, since the distributed LSM
priority queue does not provide any guarantees on the ordering of
items created by other threads.

\paragraph{Space Efficiency}

To ensure better space bounds, \texttt{spy} can be restricted to never
copy more than $k$ items (not shown in code for simplicity). Since
\texttt{spy} can only be called if the thread-local LSM is empty, this
will ensure that no more than $O(n + Tk)$ space is used, where $n$ is
the number of unique items in the distributed LSM and $T$ the number of
threads. Note that such a restriction will implicitly happen
for the combined data-structure described in the next section, since
no blocks with a level larger than $\log_2 k - 1$ will exist to spy
from, and each spied block must have a level smaller than the
previously spied block.

\subsection{Combining both priority queues}
\label{sec:combine}
We now combine the shared $k$-LSM priority queue and the distributed
LSM priority queue into a single data structure, as shown in
Listing~\ref{lst:klsm}. This priority queue consists of one
distributed LSM priority queue per thread, a single shared $k$-LSM
priority queue, and a victim array used for spying.

The \texttt{insert} method for the combined data structure simply
creates an item and inserts it into the distributed LSM priority queue
using a modified version of its \texttt{insert} method. In this
modified version, if after performing all necessary merges the
resulting block is bigger than a certain size, the block will be
inserted into the shared $k$-LSM priority queue instead of being added
to the thread-local \texttt{DistLSM}.

In the \texttt{delete\_min} method both priority queues are checked
for the minimal item, and the minimum of both marked as deleted using
an atomic test-and-set on the flag of the item returned. If the found
item was already deleted by another thread, this process is
repeated. If both are empty, \texttt{spy} is called and the operation
repeated if \texttt{spy} was successful.

\begin{algorithm}
\begin{lstlisting}[mathescape=true,columns=flexible,escapechar=---]
struct KLSM {
	thread_local DistLSM dist;
	SharedKLSM shared;
	DistLSM* victims[NUM_THREADS]

	void insert(Key key) {
		Item* item = new Item(key);
		dist.insert(item, &shared);
	}

	bool delete_min(Key& key) {
		do {
			Item* item;

			do {
				item = dist.find_min();
				if (item == nullptr) item = shared.find_min();
				else {
					Item* tmp = shared.find_min();
					if(tmp != nullptr && tmp->key < item->key)
						item = tmp;
				}
			// Try marking as deleted and repeat otherwise
			} while (item != nullptr &&
				(item->flag || !item->flag.test_and_set()));

			if (item != nullptr) return item;

		// Try finding items at other threads
		} while (spy(victims[rand(0,NUM_THREADS-1)]));
		return nullptr;
}};
\end{lstlisting}
\caption{The $k$-LSM priority queue.}
\label{lst:klsm}
\end{algorithm}

\subsection{Memory Management}
\label{sec:mem_man}

So far we ignored the issue of memory management to simplify the
presentation of the algorithm. In systems with garbage collection, we
only have to clean out dangling pointers from the LSM arrays if
\texttt{filled} is decremented to ensure garbage collection of the
given items.

In systems without garbage collection, like in our \texttt{C++}
implementations, we have to manage memory manually. We use the
wait-free memory reuse scheme by Wimmer~\cite{Wimmer2014-ct} to manage
\texttt{Item} instances. Since the scheme is not ABA safe, we change
the \texttt{flag} variable in \texttt{Item} to an integer, which
allows items to be marked as deleted in an ABA-safe manner by
incrementing \texttt{flag} with an atomic compare-and-swap. Blocks
store the expected \texttt{flag} value together with each pointer to
\texttt{Item}.

It is guaranteed that no thread will need more than four instances of
\texttt{Block} per level at any point in time, which will be allocated
on first access. In \texttt{spy}, blocks need to be verified after
copying all items to ensure ABA safety.

Two instances of \texttt{BlockArray} per thread are sufficient. For
ABA safety a version number is maintained for each instance. We
allocate these instances aligned to 2048-Byte boundaries, allowing us
to steal the ten least significant bits of a pointer to
\texttt{BlockArray}, and work around the ABA problem by stamping the
pointer with a truncated version number. The full version number is
used for verification most of the time, and only a compare-and-swap on
\texttt{shared} needs to rely on the truncated value. To reduce the
risk of wraparounds of the truncated value resulting in an ABA
problem, the full version numbers are verified directly before each
compare-and-swap to minimize the time-frame in which a wrap-around may
occur. Ten bits, and minimizing the time-frame are more than
sufficient to make this safe in practice.

\subsection{Extensions}
\label{sec:extensions}

The \emph{decrease-key} operation is an often important operation on
priority queues, but is hard to implement in parallel. It can be
worked around in a non-linearizable way by deleting a key and
reinserting it with its new value. But even deletions can be hard to
accomplish, since this either requires keys to be unique (a
requirement we do not place on our priority queue), or to maintain a
unique id or pointer to each item that is exposed to applications. In
addition, not all types of priority queues are suitable for efficient
random deletions of keys. We resolve this issue by a lazy deletion
scheme similar to the one presented by Wimmer et
al.~\cite{Wimmer2014-ym}: the priority queue can query whether an item
needs to be deleted. This can be performed whenever it is convenient
for the priority queue, which for the LSM is whenever items are copied
into a new block (deleted items do not need to be copied) or when
checking whether a block can be shrunk (here deletion happens for
items at the end of a block). We use this lazy deletion technique in
the implementation of our SSSP benchmark application in
Section~\ref{sec:eval}.

Another common operation is \emph{meld}, where two priority queues are
merged into a single one. This is also easy to provide on LSM data
structures, since merging lies at the heart of the LSM idea. However,
it is hard to provide this in a linearizable fashion. We leave this
problem for future work.

\section{Correctness}
\label{sec:correctness}

In this section we prove the linearizability of the operations on our
data structure according to structurally $\rho$-relaxed and local
ordering semantics. In addition, we show that our priority queue
fulfils the lock-free progress guarantee.

\begin{lemma}
The \texttt{insert} operation is linearizable.
\label{lemma:insert_lin}
\end{lemma}
\begin{proof}
An \texttt{insert} operation that leads to a \texttt{Block} being
published in the shared $k$-LSM priority queue is linearized at the
linearization point of \texttt{push\_snapshot}, which is the
compare-and-swap on \texttt{shared}. An \texttt{insert} operation,
which inserts the \texttt{Block} with the inserted item into the
distributed LSM priority queue is linearized when the block (which is
potentially a result of multiple merges) is written into the array
\texttt{blocks}.  In both cases, the inserted key becomes reachable
for all threads at the linearization point, and stays reachable at
least until it is marked as deleted.
\end{proof}

\begin{lemma}
The \texttt{try\_delete\_min} operation is linearizable with
structural $\rho$-relaxation semantics, where $\rho=Tk$, $T$ being the
number of participating threads and $k$ the relaxation configuration
parameter. It also fulfils local ordering semantics.
\label{lemma:del_lin}
\end{lemma}
\begin{proof}
A successful \texttt{try\_delete\_min} is linearized at the point when
the private snapshot of \texttt{shared} is verified by comparing
\texttt{shared} with \texttt{observed} for the last time. Afterwards,
both the thread-local \texttt{DistLSM} and the private snapshot of
\texttt{shared} are treated as snapshots and the structurally
$\rho$-relaxed minimum is searched for. Any inconsistencies would lead
to failures and thus repetitions, including another comparison of the
two arrays.

If an item is successfully marked as deleted it is guaranteed to have
one of the $\rho=Tk$ smallest keys at the time the snapshot of
\texttt{shared} was taken. At most $(T-1)k$ keys from the thread-local
\texttt{DistLSM}s of other threads can be skipped, since the size of a
\texttt{DistLSM} is bounded by $k$. Another $k$ keys can be ignored
due to the uniformly random selection out of at most $k+1$ keys in
\texttt{SharedKLSM::find\_min}.

An unsuccessful \texttt{try\_delete\_min} (which includes spurious
failures) is linearized at the point when \texttt{shared} is found to
be a null-pointer after the thread-local \texttt{DistLSM} has been
found empty. At this point at most $(T-1)k$ keys can be skipped, which
is the maximum number of keys that can be stored in total in the
thread-local \texttt{DistLSM}s of other threads.  Since \texttt{spy}
is allowed to spuriously fail it does not contribute to the
linearization point of a failure even though it happens after
\texttt{shared} is found to be a null-pointer.

Local ordering semantics are fulfilled, since the thread-local
\texttt{DistLSM} is always checked for the minimal key, and also all
blocks in \texttt{SharedKLSM}, which might contain items added by the
given thread.
\end{proof}

\begin{lemma}
The \texttt{insert} operation is lock-free.
\label{lemma:insert_lf}
\end{lemma}
\begin{proof}
An \texttt{insert} that only operates on the thread-local
\texttt{DistLSM} is wait-free, since no other thread can block
progress. An insertion on the shared $k$-priority queue may fail to
update \texttt{shared} so that the insertion has to be restarted with
a new snapshot, but only if another thread succeeded in modifying
it. Since at least one thread has to make progress for another thread
to fail updating the shared array, the algorithm is lock-free.
\end{proof}

\begin{lemma}
The \texttt{delete\_min} operation is lock-free.
\label{lemma:del_lf}
\end{lemma}
\begin{proof}
The \texttt{find\_min} methods of both priority queues are wait-free,
since they operate on thread-local snapshots that are not influenced
by the operations of other threads.

While the \texttt{delete\_min} operation might need more than one
attempt to delete an item even in a quiescent state, due to logically
deleted items not being physically removed, subsequent clean-up
operations on such failures ensure that an active item is eventually
found, or the priority queue empty. Other threads may increase the
number of iterations until success by deleting additional items, but
this means that another thread made progress, and thus
\texttt{delete\_min} is lock-free.
\end{proof}

\section{Evaluation}
\label{sec:eval}

We now compare our lock-free $k$-priority queue to the recent, relaxed
lock-free priority queues by Alistarh et al.~\cite{alistarhspraylist}
and Wimmer et al.~\cite{Wimmer2014-ym}, and the randomized lock-based
Multi-Queues by Rihani et al.~\cite{Sanders14}. We show that our data
structure has high absolute performance, due to its cache-efficient
layout and small amount of synchronization operations, and can provide
good scalability.

Since we focus on priority queues with relaxed semantics we cannot
show a full comparison to current state-of-the-art, non-relaxed
priority queues. Instead, we picked the lock-free priority queue by
Lind\'en and Jonsson~\cite{linden2013skiplist} as a representative of
such priority queues. We note again that scalability of exact priority
queues is limited. While it is often possible to achieve good
scalability on insertions~\cite{linden2013skiplist,LiuSpear12},
deleting the minimal key is inherently sequential. So for deletions, a
perfectly scalable exact priority queue can at best keep the absolute
throughput constant with the number of threads.

For comparison with the priority queue by Lind\'en and Jonsson, as
well as the SprayList and Multi-Queues we rely on a \emph{throughput}
benchmark, which lets all threads randomly insert and delete keys from
a priority queue that is \emph{prefilled} with a given number of
keys. This benchmark creates extremely high contention on the priority
queue and thus gives good insight into the practical scalability a
priority queue. The same benchmark was also used by Alistarh et
al.~\cite{alistarhspraylist} and Rihani et al.\cite{Sanders14} to
evaluate their data structures and thus facilitates direct
comparability.  In the experiments the ratio between insertions and
deletions is 50-50.

The data structures by Wimmer et al.~\cite{Wimmer2014-ym} are part of
a task scheduling system, and cannot be used as standalone data
structures. This makes it impossible to do a fair comparison with
these data structures with the throughput benchmark. Instead, we
adapted the \emph{SSSP benchmark} of~\cite{Wimmer2014-ym} to evaluate
their priority queues. This benchmark is a label-correcting version of
Dijkstra's single-source shortest path algorithm, which is
parallelized in a straightforward manner using a concurrent priority
queue. It uses a lazy deletion scheme in connection with reinsertion
of keys instead of an explicit \emph{decrease-key} operation. We
adapted our priority queue to support this lazy deletion scheme as
described in Section~\ref{sec:extensions}. Since this extension has a
significant impact on performance, but is not available in the other
priority queues, and since there was no easy way to directly delete
keys from these priority queues with the available implementations
either, we decided not to include these data structures in the SSSP
benchmark.

\paragraph{Environment}

We have implemented our data structure in \texttt{C++}. We relaxed the
memory consistency for operations on synchronization wherever possible
to reduce the synchronization cost. We used manual memory management
for our priority queue as described in Section~\ref{sec:mem_man}. Our
priority queue implementation, as well as the benchmarks used in this
paper are available for download as part of the Pheet
framework\footnote{\url{www.pheet.org}}.

All experiments were performed on an 80-core Intel Xeon based system
comprised of eight Intel Xeon E7-8850 processors of 10 cores each.
The system has 1TB of main memory. The experiments were run under
Debian Linux and were compiled using \texttt{gcc 4.9.1}. We verified
our findings using a 48-core AMD Opteron based machine. Due to space
considerations, and since the Intel system is closer to the machines
used by Alistarh et al., Sanders et al. and Wimmer et al., we 
give results only for the Intel system.

\paragraph{Methodology}
The throughput benchmark was run for $10$ seconds for each experiment,
and the average throughput per second is shown (for a 50-50 mix of
insertions and deletions). SSSP was run on Erd\H{o}s-R\'enyi random
graphs with $10000$ nodes and edge probability $50\%$; edge weights
are randomly chosen integers in the range $[1,100000000]$. All
experiments were repeated for 30 times. Mean values with confidence
intervals are reported in the plots.

\subsection{Results}

\begin{figure*}[t]
\begin{center}
\includegraphics{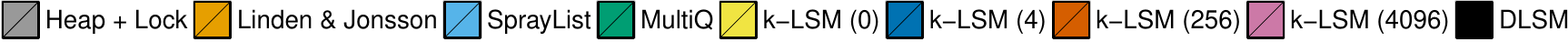}
\end{center}
%\subfloat[][]{\includegraphics{measurements/graph-p0001.pdf}}
\includegraphics{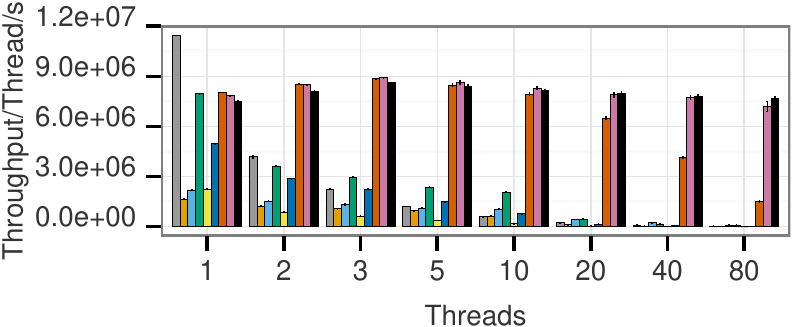}
\hfill
\includegraphics{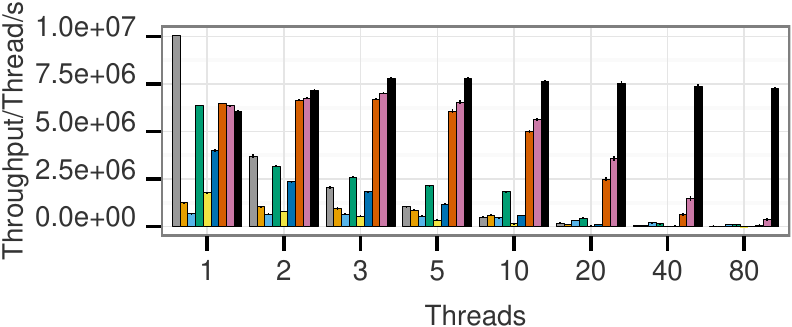}
\caption{Throughput per Thread per second for priority queues
  prefilled with $10^6$ (left) and $10^7$ (right) elements.  We use
  \emph{throughput per thread} for better readability, and thus a
  constant throughput in the plot corresponds to a linear speedup.}
\label{fig:throughput}
\end{figure*}

Results for the throughput benchmark can be seen in
Figure~\ref{fig:throughput}. The benchmarked implementations are a
binary heap protected by a spin-lock, the skiplist-based priority
queue by Lind\'en and Jonsson, the SprayList, the Multi-Queue by
Rihani et al., our priority queue with different values for $k$, and
the distributed LSM priority queue (DLSM), which is our priority queue
without $\rho$-relaxation guarantees.

We first observe that in a sequential setting, the performance of the
DLSM is close to the binary heap, and due to the local ordering
semantics also provides the same guarantees. Our priority queue with
$k=0$ is significantly slower due to the high overhead of maintaining
the shared $k$-LSM priority queue. We note that the Lind\'en and
Jonsson priority queue is slightly slower than $k$-LSM with $k=0$ due
to the worse cache-efficiency of skip-lists and the higher number of
required synchronization operations per insertion and deletion.

The Lind\'en and Jonsson priority queue scales well for a non-relaxed
priority queue and can outperform the binary heap with the global
lock, as well as our priority queue with $k=0$. The performance and
scalability of our data structure improves drastically with higher
$k$, making it easy to beat non-relaxed priority queues. Our DLSM even
scales superlinearly, due to the smaller number of keys stored per
thread as the number of threads increase.

The experiments for the SprayList were run with the default settings
of the implementation. We had to increase the number of allocators in
the custom allocator to $10$, otherwise the algorithm ran out of
memory. We were only able to run the benchmark for one second (instead
of ten), again due to the algorithm running out of memory. The
SprayList has better scalability than the Lind\'en and Jonsson
priority queue. While it seems to scale better than our priority queue
with $k=256$, its absolute performance is worse, except for the case
of $10^7$ elements and $80$ threads.

As for the relaxation, the SprayList will return one of the $O(T
\log_2^3 T)$ smallest keys whp.\ as opposed to our priority queue,
which will return one of the $Tk$ smallest keys on
\texttt{delete\_min}. The exact constant factors of the SprayList are
not documented, making a direct comparison with our priority queue
difficult. To give an example, for 64 threads, the SprayList will have
$\rho = 13824c$ for some constant $c>1$, whereas our priority queue
has $\rho = 64k$ for a configurable $k$ (e.g. for $k=256$,
$\rho=16384$). Note that for $k$-LSM $\rho$ is the worst-case
bound. For SprayList the worst-case cannot be bounded, since the
SprayList can be arbitrarily modified while the \texttt{delete\_min}
operation traverses the list. The SprayList also does not provide
local ordering semantics.

The Multi-Queue experiments were run with the
Boost\footnote{\url{www.boost.org}} 8-ary heap enabled and a setting
of $c=2$, which means that $2T$ priority queues are used. According to
the inventors of the Multi-Queue, the expected quality should roughly
correspond to the $k$-LSM priority queue with $k=4$. The Multi-Queue
exhibits very good performance and scalability. Unfortunately no
worst-case bounds on the quality can be given for the Multi-Queue's
\emph{delete-min} operation, since a stalling thread can block
delete-min access to an arbitrary number of keys.

Comparing the results for different prefill values, we observe that
our implementation performs better for a fixed $k$ when less elements
are stored in the priority queue. This comes from the fact that a
larger amount of items in the priority queue increases the probability
that the selected minimal key in the shared $k$-LSM is smaller than
the minimal key out of the thread-local \texttt{DistLSM} for this
benchmark. Thus the limited scalability of the shared $k$-LSM priority
queue becomes more dominant.

\begin{figure*}[t]
\begin{center}
\includegraphics{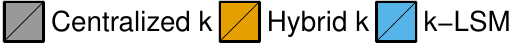}
\end{center}
%\subfloat[][]{\includegraphics{measurements/graph-p0001.pdf}}
\includegraphics{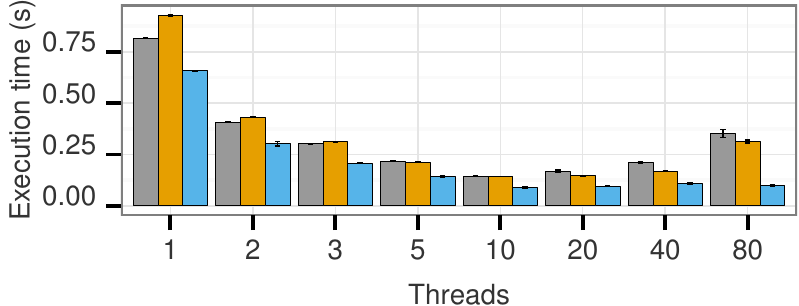}
\hfill
\includegraphics{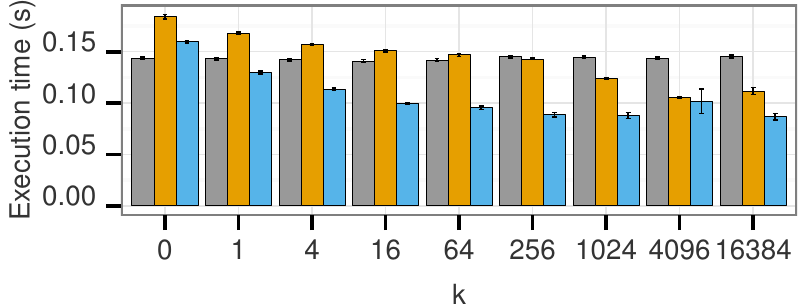}
\caption{Execution times for SSSP benchmark for varying numbers of threads ($k=256$) and values for $k$ ($10$ threads).}
\label{fig:sssp}
\end{figure*}

Results for the SSSP benchmark are shown in Figure~\ref{fig:sssp}. We
compare our priority queue ($k$-LSM) against the centralized
$k$-priority queue and hybrid $k$-priority queues by Wimmer et
al.~\cite{Wimmer2014-ym}. What can be seen here is that while the
algorithm used by the benchmark seems to have limited scalability so
that it does not scale to more than $10$ threads, our priority queue
provides stable performance while the other priority queues experience
a slowdown with more threads.

We also explore the influence of the relaxation parameter $k$ on the
performance of the algorithm. For this we fixed the number of threads
to $10$ and varied the parameter $k$. For the centralized $k$-priority
queue no visible difference between different values for $k$ can be
seen, making it the best priority queue for $k=0$, but not the best of
all. For the other priority queues there exists a best value, where
the additional work that needs to be performed due to the relaxation
does not outweigh the scalability gains. For the hybrid $k$-priority
queue this value is somewhere around $k=4096$, where on average $305$
additional iterations needed to be performed compared to a sequential
execution (not shown in graphs). For our priority queue we found such
a best value at $k=256$ (+$362$ iterations), but also had good results
with $k=16384$ (+$3965$ iterations).

\section{Conclusion}

We presented a novel, relaxed concurrent priority queue with lock-free
progress guarantees. In a sequential setting, its performance is
comparable to a binary heap, and exhibits very good scalability for a
reasonably large value for the relaxation parameter $k$. Performance-
and scalability-wise it is competitive to other, recent, lock-free
relaxed priority queues. In particular, it seems superior to the
SprayList~\cite{alistarhspraylist} in most settings, and has the
advantage of providing fixed and easy to calculate relaxation
guarantees and local ordering semantics.  The sensitivity of our
priority queue to the number of queued elements makes the SprayList
the better choice for settings with a huge number of elements and high
contention. In comparison with the priority queues by Wimmer et
al.~\cite{Wimmer2014-ym}, we provide both better scalability and
absolute performance for the SSSP application.

The distributed LSM priority queue has high scalability, except for
spying for which the cost increases with the number of items that are
spied. We intend to resolve this by allowing spy operations to copy
read-only references to blocks. However, the main scalability
bottleneck of our priority queue is the shared $k$-LSM priority queue.
We see lots of potential for improvement there, utilizing techniques
like lazy merging and helping schemes.

The appealing Multi-Queue by Rihani et al.~\cite{Sanders14} seems to
provide better trade-offs between expected quality and scalability,
but cannot provide any worst-case guarantees. We hope to collaborate
with Rihani et al.  to devise a relaxed priority queue that combines
the best of the two ideas.

Finally, an interesting area to explore is priority queues with
elimination and batching as presented by Calciu et
al.~\cite{Calciu14}. We believe that the LSM priority queue is
well-suited for batching techniques and it might be possible to
achieve scalability for our data structure using these techniques
instead of relaxation.

\bibliographystyle{abbrv}
\bibliography{klsm}

\end{document}